\documentclass[a4paper,UKenglish,cleveref, autoref, thm-restate]{lipics-v2021}

\pdfoutput=1 
\hideLIPIcs

\title{Fixed-Template Promise Model Checking Problems}

\author{Kristina Asimi}{Department of Algebra, Faculty of Mathematics and Physics, Charles University, Czechia \and \url{}}{asimptota94@gmail.com}{}{}

\author{Libor Barto}{Department of Algebra, Faculty of Mathematics and Physics, Charles University, Czechia \and \url{https://www2.karlin.mff.cuni.cz/~barto/} }{libor.barto@gmail.com}{https://orcid.org/0000-0002-8481-6458}{}

\author{Silvia Butti}{Department of Information and Communication Technologies, Universitat Pompeu Fabra, Spain \and \url{https://sites.google.com/view/silviabutti/}}{silvia.butti@upf.edu}{https://orcid.org/0000-0002-0171-2021}{}

\authorrunning{K. Asimi, L. Barto, and S. Butti}

\Copyright{Kristina Asimi, Libor Barto, and Silvia Butti}



\category{} 

\relatedversion{}

\funding{Kristina Asimi and Libor Barto have received funding from the European Research Council (ERC) under the European Unions Horizon 2020 research and innovation programme (Grant Agreement No.~771005, CoCoSym). Silvia Butti was supported by a MICCIN grant PID2019-109137GB-C22 and by a fellowship from ``la Caixa'' Foundation (ID 100010434). The fellowship code is LCF/BQ/DI18/11660056. This project has received funding from the European Union’s Horizon 2020 research and innovation programme under the Marie Skłodowska-Curie grant agreement No. 713673.}

\nolinenumbers

\usepackage{makecell}

\newcommand{\rel}[1]{\mathbb{#1}}
\newcommand{\tuple}[1]{\mathbf{#1}}
\newcommand{\AB}{(\rel A,\rel B)}
\newcommand{\CD}{(\rel C,\rel D)}
\newcommand{\ar}{\mathrm{ar}}
\newcommand{\compl}[1]{\overline{#1}}
\newcommand{\eval}{\varepsilon}
\newcommand{\paragr}[1]{\noindent\textbf{#1}}

\newcommand{\powerset}{\mathcal{P}_{\neq \emptyset}}

\newcommand{\logic}{\mathcal{L}}
\newcommand{\connectives}{\{\exists, \forall, \land, \lor, =, \neq, \neg\}}

\newcommand{\yes}{\mathtt{Yes}}
\newcommand{\no}{\mathtt{No}}

\newcommand{\Lspace}{\mathrm{L}}
\newcommand{\Ptime}{\mathrm{P}}
\newcommand{\NP}{\mathrm{NP}}
\newcommand{\coNP}{\mathrm{coNP}}
\newcommand{\PSPACE}{\mathrm{PSPACE}}

\newcommand{\PMC}{\mathrm{PMC}}
\newcommand{\MC}{\mathrm{MC}}
\newcommand{\LPMC}{\logic\mbox{-}\mathrm{PMC}}
\newcommand{\LMC}{\logic\mbox{-}\mathrm{MC}}
\newcommand{\ea}{\ensuremath{\{\exists,\land\}}}
\newcommand{\eaeq}{\ensuremath{\{\exists,\land, =\}}}
\newcommand{\eao}{\ensuremath{\{\exists,\land,\lor\}}}
\newcommand{\eaoeq}{\ensuremath{\{\exists,\land,\lor, =\}}}
\newcommand{\efa}{\ensuremath{\{\exists,\forall,\land\}}}
\newcommand{\efao}{\ensuremath{\{\exists,\forall,\land,\lor\}}}  \newcommand{\efaoeq}{\ensuremath{\{\exists,\forall,\land,\lor,=\}}}  \newcommand{\efaoneq}{\ensuremath{\{\exists,\forall,\land,\lor,\neq\}}}  \newcommand{\efaoneg}{\ensuremath{\{\exists,\forall,\land,\lor,\neg\}}}  
\newcommand{\efaeq}{\ensuremath{\{\exists,\forall,\land,=\}}}

\newcommand{\MultHom}{\mathrm{MuHom}}
\newcommand{\SMultHom}{\mathrm{SMuHom}}

\newcommand{\smulthom}{\mathrm{smuhom}}

\newcommand{\ashop}{\rotatebox[origin=c]{180}{$\forall$}\text{-}\smulthom}
\newcommand{\eshop}{\rotatebox[origin=c]{180}{$\exists$}\text{-}\smulthom}
\newcommand{\aeshop}{\rotatebox[origin=c]{180}{$\exists\forall$}\mbox{-}\smulthom}

\newcommand{\CSP}{\mathrm{CSP}}

\newcommand{\nae}{\textnormal{NAE}}

\newcommand{\rb}{\textnormal{Rb}}

\begin{document}

\maketitle

\begin{abstract}
    The fixed-template constraint satisfaction problem (CSP) can be seen as the problem of deciding whether a given  primitive positive first-order sentence is true in a fixed structure (also called model).
    We study a class of problems that generalizes the CSP simultaneously in two directions: we fix a set $\mathcal{L}$ of quantifiers and Boolean connectives, and  we specify two versions of each constraint, one strong and one weak. Given a sentence which only uses symbols from $\mathcal{L}$, the task is to distinguish whether the sentence is true in the strong sense, or it is false even in the weak sense.

We classify the computational complexity of these problems for the existential positive equality-free fragment of first-order logic, i.e., $\mathcal{L} = \{\exists,\land,\lor\}$, and we prove some upper and lower bounds for the positive equality-free fragment, $\mathcal{L} = \{\exists,\forall,\land,\lor\}$. The partial results are sufficient, e.g., for all extensions of the latter fragment. 
\end{abstract}

\section{Introduction}

The fixed-template finite-domain constraint satisfaction problem (CSP) is a framework for expressing many computational problems such as various versions of logical satisfiability, graph coloring, and systems of equations, see the survey~\cite{Barto2017polymorphisms}. A convenient formalization, that we adopt in this paper, is as follows: a \emph{template} is a relational structure $\rel A$, and the  \emph{CSP over $\rel A$} is the problem of deciding whether a given $\{\exists,\land\}$-sentence is true in $\rel A$. Here, an \emph{$\{\exists,\land\}$-sentence} is a sentence of first-order logic that uses only the relation symbols of $\rel A$, the logical connective $\land$, and the quantifier $\exists$. To see that this formalization indeed expresses \emph{constraint} satisfaction problems, consider, e.g., the sentence $\exists x \exists y \exists z \ R(x,y) \wedge S(y,z)$: this sentence is true in a structure $\rel A$ if the variables $x,y,z$ can be evaluated so that both atomic formulas (constraints) are satisfied in $\rel A$. 

Motivated by recent developments in the area, we study an extension of this framework in two simultaneous directions. One direction, discussed in Subsection~\ref{subsec:intro-mc},  is to enable other choices of permitted quantifiers and connectives. Another direction, discussed in Subsection~\ref{subsec:intro-pmc}, is to consider two versions of each relation, strong and weak (a so-called promise problem). Our contributions are then described in Subsection~\ref{subsec:intro-contrib}.

\subsection{Model checking problem parametrized by the model} \label{subsec:intro-mc}

The model checking problem~\cite{MadelaineTetrachotomyJournal} takes as input a structure $\rel A$ (often called a model) and a sentence $\phi$ in a specified logic and asks whether $\rel A \vDash \phi$, i.e., whether $\rel A$ satisfies $\phi$. We study the situation where $\rel A$ is a fixed finite relational structure, so the input is simply $\phi$, and the logic is a fragment of the first-order logic obtained by restricting the allowed quantifiers to a subset $\logic$ of $\connectives$. Thus, for each $\rel A$ and each of the $2^7$ choices for $\logic$, we obtain a computational problem, which we call the \emph{$\logic$-Model Checking Problem over $\rel A$} and denote  $\LMC(\rel A)$. 

The computational complexity classification of $\{\exists,\wedge\}$-$\MC(\rel A)$, i.e., $\CSP$ over $\rel A$, has been a very active research program in the last 20 years, which culminated in the celebrated dichotomy theorem obtained independently in \cite{Bulatov2017} and \cite{Zhu20}: each $\CSP$ over $\rel A$ is  in $\Ptime$ (solvable in polynomial time) or is $\NP$-complete. 
For the case $\logic = \{\exists, \forall, \land\}$, $\LMC(\rel A)$ is the so called \emph{quantified CSP}, another well-studied class of problems, see the survey~\cite{QCSPsurvey}. It was widely believed that this class exhibits a $\Ptime$/$\NP$-complete/$\PSPACE$-complete trichotomy~\cite{Chen12meditations}. A recent breakthrough~\cite{ZhukM20} shows that at least three more complexity classes appear within quantified CSPs, and ongoing work suggests that even 6 is not the final number. In any case, the full complexity classification of $\{\exists,\forall,\land\}$-$\MC(\rel A)$ is a challenging open problem. 

The remaining $2^7-2$ choices for $\logic$ do not need to be considered separately. For instance, $\{\exists,\land,=\}$-$\MC(\rel A)$ is no harder than $\{\exists,\land\}$-$\MC(\rel A)$ because equalities can be propagated out in this case, and $\{\forall,\lor\}$-$\MC(\rel A)$ is dual to $\{\exists,\land\}$-$\MC(\rel A)$ so we get a $\Ptime$/$\coNP$-complete dichotomy for free, etc. Moreover, some choices of $\logic$, such as $\logic = \{\exists, \lor\}$, lead to very simple problems. It turns out~\cite{Martin2008modelchecking} (see Subsection~\ref{subsec:interesting-fragments}) that, in addition to $\logic=\{\exists,\land\}$ and $\logic=\efa$, only two more fragments need to be considered in order to fully understand the complexity of $\LMC(\rel A)$, namely $\logic=\eao$ and $\logic = \efao$.

The former fragment was addressed in~\cite{Martin2008modelchecking}: except for a simple case solvable in polynomial time (in fact, $\Lspace$, the logarithmic space), all the remaining problems are $\NP$-complete. The latter fragment turned out to be more challenging but, after a series of partial results~\cite{Martin2008modelchecking,madelaine2012complexity,MartinMartin4element} (see also~\cite{MartinLattice,Carvalho2021}), the full complexity classification was given in~\cite{Madelaine2011tetrachotomy,MadelaineTetrachotomyJournal}: each problem in this class is  in $\Ptime$ (even $\Lspace$), or is $\NP$-complete, $\coNP$-complete, or $\PSPACE$-complete. These results are summarized in Figure~\ref{fig:mc}.

\begin{figure}[h] 
\begin{center}
\begin{tabular}{ |c|c| } 
 \hline
 $\LMC(\rel A)$  & Complexity\\ 
 \hline \hline
 $\{\exists,\land\}$-$\MC(\rel A)$ (CSP) & dichotomy: $\Ptime$ or $\NP$-complete \\ 
 \hline
 $\{\exists,\forall, \land\}$-$\MC(\rel A)$ (QCSP)& $\geq 6$ classes\\
 \hline
 $\{\exists,\land,\lor\}$-$\MC(\rel A)$ & dichotomy: $\Lspace$ or $\NP$-complete\\
 \hline
 $\{\forall,\exists,\land,\lor\}$-$\MC(\rel A)$ & tetrachotomy: $\Lspace$, $\NP$-complete,
  $\coNP$-complete, $\PSPACE$-complete\\
  \hline
\end{tabular}
\caption{\label{fig:mc} Known complexity results for $\logic$-$\MC(\rel A)$.}
\end{center}
\end{figure}

\subsection{Promise model checking problem} \label{subsec:intro-pmc}

The Promise CSP is a recently introduced extension of the CSP framework motivated by open problems in (in)approximability of satisfiability and coloring problems~\cite{AGH17,BG18,Barto2021algebraicPCSP}. The template consists of two structures $\rel A$ and $\rel B$ of the same signature, where $\rel A$ specifies a strong form of each relation and $\rel B$ its weak form. The Promise CSP over $(\rel A,\rel B)$ is then the problem of distinguishing $\ea$-sentences that are true in $\rel A$ from those that are not true in $\rel B$.

For example, by choosing an appropriate template, we obtain the problem of distinguishing $k$-colorable graphs from those that are not even $l$-colorable (where $k\leq l$ are fixed), a problem whose complexity is notoriously open.  

The generalization of Promise CSP over $\AB$ to an arbitrary choice $\logic \subseteq \connectives$ is referred to as the \emph{$\logic$-Promise Model Checking Problem over $\AB$} and is denoted $\LPMC \AB$. Similarly as in the special case $\rel A=\rel B$, which is exactly $\LMC (\rel A)$, it is sufficient to consider only four fragments.
A full complexity classification for $\ea$-$\PMC$ (i.e., Promise CSP) is much desired but widely open, and $\efa$-$\PMC$ is likely even harder. This work concentrates on the remaining two classes of problems, $\eao$-$\PMC$ and $\efao$-$\PMC$. 

Our motivation was that these cases might be substantially simpler, as indicated by the non-promise special case, and at the same time, the investigation could uncover interesting intermediate problems towards the grand endeavor of understanding the sources of tractability and hardness in computation. We believe that our findings confirm this hope.

\begin{example}

Consider structures $\rel A$ and $\rel B$ with a single relation symbol $=$ interpreted as the equality on a three-element domain in $\rel A$ and as the equality on a two-element domain in $\rel B$. For $\logic=\efao$, both $\logic$-$\MC(\rel A)$ and $\logic$-$\MC(\rel B)$ are $\PSPACE$-complete problems, see~~\cite{Martin2008modelchecking}.

It is not hard to see that every $\logic$-sentence true in $\rel A$ is also true in $\rel B$. In this sense, the relation in $\rel A$ is stronger than the relation in $\rel B$. On the other hand, there are $\logic$-sentences true in $\rel B$ that are not true in $\rel A$, e.g., $\phi =
\forall x \exists y \forall z \ (z=x) \lor (z=y)$.
Therefore, $\logic$-$\PMC \AB$ could potentially be easier than the above non-promise problems -- instances such as $\phi$ need not be considered (there is no requirement on the algorithm for such inputs). Nevertheless, the problem remains $\PSPACE$-complete, as shown in Proposition~\ref{prop:hardness_of_equality}.
\end{example}

\subsection{Contributions} \label{subsec:intro-contrib}

Theorem~\ref{thm:eao_definability} and Theorem~\ref{thm:efao_definability} provide basics for an algebraic approach to $\eao$-$\PMC$ and $\efao$-$\PMC$ by characterizing definability in terms of compatible functions: multi-homomorphisms for the $\eao$ fragment and surjective multi-homomorphisms (\emph{smuhoms}) for $\efao$. The proofs can be obtained as relatively straightforward generalizations of the proofs for $\MC$ in~\cite{MadelaineTetrachotomyJournal}; however, we believe that our approach is somewhat more transparent. In particular, it allows us to easily characterize meaningful templates for these problems (Propositions~\ref{prop:existential_templates} and \ref{prop:positive_templates}).

For $\eao$-$\PMC$, we obtain an $\Lspace$/$\NP$-complete dichotomy in Theorem~\ref{thm:eao_classification}. It turns out that, apart from some simple cases, the problem is $\NP$-complete. Interestingly, there is a ``single reason'' for hardness: the NP-hardness of coloring a rainbow colorable hypergraph from~\cite{Guruswami2018}.

For $\efao$-$\PMC$, our complexity results  are only partial, leaving two gaps for further investigation. The results are sufficient for full complexity classification of $\LPMC \AB$ in the case that $\logic=\efao$ and one of the structures $\rel A$, $\rel B$ has a two-element domain, and also in the case that $\logic \supsetneq \efao$. We also give some examples where our efforts have failed so far. One such example is a particularly interesting $\efao$-$\PMC$ over 3-element domains:  given a $\efao$-sentence $\phi$ whose atomic formulas are all of the form $R^i(x)$, $i \in \{1,2,3\}$, distinguish between the case where $\phi$ is true when $R^i(x)$ is interpreted as ``$x = i$'', and the case where $\phi$ is false when $R^i(x)$ is interpreted as ``$x \neq i$''.

Our complexity results are summarized in Figure~\ref{fig:result}, the conditions for $\logic=\efao$ are stated in terms of special surjective multi-homomorphisms of the template, introduced in Subsection~\ref{subsec:positive-membership}.

\begin{figure}[h] 
\begin{center}
\begin{tabular}{ |c|c|c| } 
\hline
 \textbf{$\LPMC\AB$} & \textbf{Condition} & \textbf{Complexity} \\ 
 \hline\hline
 $\efa$-$\PMC \AB$ & & $\Lspace$/$\NP$-complete \\
 \hline
 &  \makecell{ $\aeshop$, or
     $\ashop$ and \\ $\eshop$ and $\rel A, \rel B$ digraphs} &
     $\Lspace$ \\
 \cline{2-3}
$\efao$-$\PMC \AB$ & $\ashop$ and $\eshop$ & $\NP \cap \coNP$\\
 \cline{2-3}
  & $\ashop$, no $\eshop$ & $\NP$-complete \\ 
 \cline{2-3}
 & $\eshop$, no $\ashop$ & $\coNP$-complete\\
\cline{2-3}
 & no $\ashop$ and no $\eshop$ & $\NP$-hard and $\coNP$-hard\\
 \hline
\makecell{$\efaoeq$-$\PMC \AB$,\\ $\efaoneq$-$\PMC \AB$,\\ $\efaoneg$-$\PMC\AB$}& 
& $\Lspace$/$\PSPACE$-complete\\
\hline
\end{tabular}
\caption{\label{fig:result} Complexity results for $\LPMC\AB$.}
\end{center}
\end{figure}

\section{Preliminaries}

\paragr{Structures.}
We use a standard model-theoretic terminology, but restrict the generality of some concepts for the purposes of this paper. 
A \emph{relation} of arity $n \geq 1$ on a set $A$ is a set of $n$-tuples of elements of $A$, i.e., a subset of $A^n$. 
The \emph{complement} of a relation $S$ is denoted $\compl{S} := A^n \setminus S$. The equality relation on $A$ is denoted $=_A$ and the disequality relation $\neq_A$. Components of a tuple $\tuple{a}$ are referred to as $a_1$, $a_2$, \dots, i.e., $\tuple{a} = (a_1, \dots, a_n)$.

A \emph{signature} is a nonempty collection of relation symbols each with an associated arity, denoted $\ar(R)$ for a relation symbol $R$. A \emph{relational structure} (also called a \emph{model}) $\rel A$ in the signature $\sigma$, or simply a \emph{structure}, consists of a finite set $A$ of size at least two, called the \emph{universe} of $\rel A$, and a nonempty proper relation $\emptyset \subsetneq R^{\rel A} \subsetneq A^{\ar(R)}$
for each symbol $R$ in $\sigma$, called the \emph{interpretation} of $R$ in $\rel A$. Two structures are called \emph{similar} if they are in the same signature. 
The \emph{complement} of a relational structure $\rel A$ is obtained by taking complements of all relations in the structure and is denoted $\compl{\rel A}$. 
A structure over a signature containing a single binary relation symbol is called a \emph{digraph}.

We emphasize that the universe of a structure is denoted by the same letter as the structure, that the universe of every structure in this paper is assumed to be finite and at least two-element, and that each relation in a structure is assumed to be at least unary, nonempty and proper. These nonstandard requirements are placed for technical convenience and do not significantly decrease the generality of our results.

Given two similar structures $\rel A$ and $\rel B$, a  function $f$ from $A$ to $B$ is called a \emph{homomorphism} from $\rel A$ to $\rel B$ if $f(\tuple{a}) \in R^{\rel B}$ for any $\tuple{a} \in R^{\rel A}$, where $f(\tuple{a})$ is computed component-wise. We only work with total functions, that is, $f(a)$ is defined for every $a \in A$. 

\smallskip

\paragr{Multi-homomorphisms.}
A \emph{multi-valued function} $f$ from $A$ to $B$ is a mapping from $A$ to $\powerset B$, the set of all nonempty subsets of $B$. It is called \emph{surjective} if for every $b \in B$, there exists $a \in A$ such that $b \in f(a)$. The \emph{inverse} of a surjective multi-valued function $f$ from $A$ to $B$ is the multi-valued function from $B$ to $A$ defined by $f^{-1}(b) = \{a: b \in f(a)\}$. 
For a tuple $\tuple{a} \in A^n$ we write $f(\tuple{a})$ for $f(a_1) \times \dots \times f(a_n)$.
The value $\max \{|f(a)|: a \in A\}$ is referred to as the \emph{multiplicity} of $f$; in particular, multi-valued functions of multiplicity one are essentially functions. 
For two multi-valued functions $f$ and $f'$ from $A$ to $B$, we say that $f'$ is \emph{contained in} $f$ if $f'(a) \subseteq f(a)$ for each $a \in A$.

Given two similar structures $\rel A$ and $\rel B$, a multi-valued function $f$ from $A$ to $B$ is called a \emph{multi-homomorphism}\footnote{
We deviate here from the terminology of~\cite{Madelaine2011tetrachotomy, madelaine2012complexity} because it would not work well in the promise setting.
} from $\rel A$ to $\rel B$ if for any $R$ in the signature and any $\tuple{a} \in R^{\rel A}$, we have $f(\tuple{a}) \subseteq R^{\rel B}$, i.e., $\tuple{b} \in R^{\rel B}$ whenever $b_i \in f(a_i)$ for each $i \in [\ar(R)] = \{1,2, \dots, \ar(R)\}$. Notice that if $f$ is a multi-homomorphism from $\rel A$ to $\rel B$, then so is any multi-valued function contained in $f$. In particular, if $f$ is a multi-homomorphism from $\rel A$ to $\rel B$, then any function $g: A \to B$ with $g(a) \in f(a)$ for each $a \in A$ is a homomorphism from $\rel A$ to $\rel B$. The converse does not hold in general, as witnessed by structures $\rel A = \rel B$ with a single binary equality relation and any multi-valued function of multiplicity greater than one.

The set of all multi-homomorphisms from $\rel A$ to $\rel B$ is denoted by
$\MultHom \AB$ and the set of all surjective multi-homomorphisms by
$\SMultHom \AB$.

\smallskip

\paragr{Fragments of first-order logic.} Let $\logic \subseteq \connectives$ and fix some signature. By an \emph{$\logic$-sentence} (resp., \emph{$\logic$-formula}) we mean a sentence (resp., formula) of first-order logic that only uses variables (denoted $x_i$, $y_i$, $z_i$), relation symbols in the signature, and connectives and quantifiers in $\logic$. We refer to this fragment of first-order logic  as the \emph{$\logic$-logic}.

The \emph{prenex normal form} of an $\logic$-formula is an equivalent formula that begins with quantified variables followed by a quantifier-free formula. The prenex normal form can be computed in logarithmic space and it is an $\logic$-formula whenever $\logic$ does not contain the negation.

For a structure $\rel A$ in the signature and an $\logic$-sentence $\phi$, we write $\rel A \vDash \phi$ if $\phi$ is satisfied in $\rel A$. 
More generally, given an $\logic$-formula $\psi$, a tuple of distinct variables $(v_1, \dots, v_n)$ which contains every free variable of $\psi$ and a tuple $(a_1, \dots, a_n) \in A^n$, we write $\rel A \vDash \psi(a_1, \dots, a_n)$  if $\psi$ is satisfied when $v_1, \dots, v_n$ are evaluated as $\eval_{A}(v_1)=a_1, \dots, \eval_{A}(v_n)=a_n$, respectively. Notice that variables $v_1, \dots, v_n$ indeed need to be pairwise distinct, otherwise this notation would not make sense. The tuple $(v_1, \dots, v_n)$ is often specified by writing $\psi = \psi(v_1, \dots, v_n)$. 

We say that a relation $S \subseteq A^n$ is \emph{$\logic$-definable} from $\rel A$ if there exists an $\logic$-formula $\psi(v_1, \dots, v_n)$
such that, for all $ (a_1, \dots, a_n) \in A^n$, we have $(a_1, \dots, a_n) \in S$ if and only if $\rel A \vDash \psi(a_1, \dots, a_n)$. In this case,
we also say that $\psi(v_1, \dots, v_n)$ defines $S$ in $\rel A$.

\section{Promise model checking} 

In this section we define the promise model checking problem restricted to $\logic \subseteq \connectives$. We start by briefly discussing the non-promise setting.

\subsection{Model checking problem} 
Let $\logic \subseteq \connectives$ and $\rel A$ be a structure in a signature $\sigma$. Recall that the \emph{$\logic$-Model Checking
Problem over $\rel A$}, denoted $\LMC(\rel A)$, is the problem of deciding whether a given $\logic$-sentence $\phi$ (in the same signature as $\rel A$) is true in $\rel A$.

A simple but important observation sometimes allows us to compare the complexity of the $\LMC$ problems over two templates $\rel A$ and $\rel C$ with the same universe $A=C$ but possibly different signatures: If every relation in $\rel C$ is $\logic$-definable from $\rel A$, then $\LMC(\rel C)$ can be reduced in polynomial-time (even logarithmic space) to $\LMC(\rel A)$. Indeed, the reduction amounts to replacing atomic formulas of the form $R(\tuple{v})$ by their definitions.

The starting point of the algebraic approach to $\LMC$ is to find a characterization of definability in terms of certain ``compatible functions'' or ``symmetries'' (so called polymorphisms for $\logic=\eaeq$~\cite{Barto2017polymorphisms}, surjective polymorphisms for $\logic = \efaeq$~\cite{QCSPsurvey}, 
multi-endomorphisms for $\logic=\eao$,
surjective multi-endomorphisms for $\logic = \efao$ ~\cite{MadelaineTetrachotomyJournal}; see also~\cite{Borner2008}). Because such characterizations are central in this paper as well, we now explain the basic idea for a simple case.

For $\logic = \eaoeq$, the appropriate type of compatible function is endomorphism: a nonempty relation $S \subseteq A^n$ is $\logic$-definable from $\rel A$ if and only if  it is invariant under every endomorphism of $\rel A$ (i.e., a homomorphism from $\rel A$ to itself). The forward direction is well-known and easy to verify. For the backward direction, assume $A = [k] := \{1, \dots, k\}$ and consider the following formula.
\begin{equation} \label{eq:endo}
  \phi(x_1, \dots, x_k) :=    \bigwedge_{R \in \sigma} \bigwedge_{\tuple{r} \in R^{\rel A}} R(x_{r_1}, \dots, x_{r_{\ar(R)}})
\end{equation}
It follows immediately from definitions that, for any structure $\rel E$ in the signature of $\rel A$, $\rel{E} \vDash \phi(e_1, \dots, e_k)$ if and only if the mapping defined by $i \mapsto e_i$ for each $i \in [k]$ is a homomorphism from $\rel A$ to $\rel E$. This in particular holds for $\rel E = \rel A$. By existential quantification we can then obtain an $\logic$-formula defining the closure of any tuple $\tuple{a} \in A^n$ with distinct entries under endomorphisms of $\rel A$; e.g., $\psi(x_1,x_3,x_2) := (\exists x_4) (\exists x_5) \dots (\exists x_k) \phi$ defines the closure of $(1,3,2)$ under endomorphisms. Using $=$ we can also define closures of the remaining tuples with repeated entries. Finally, $S$ is the union of closures of its members (since it is closed under endomorphisms of $\rel A$), so $S$ can be defined by a disjunction of  formulas that we have already found (after appropriately renaming variables). 

Notice that this construction would not work without the equality in $\logic$ because of tuples with repeated entries. This is the reason why we need to work with multi-valued functions for the equality-free logics that we deal with in this paper. 

\subsection{Promise model checking problem}
Let $\logic \subseteq \connectives$.
The $\logic$-Promise Model Checking
Problem over a pair of similar structures $\AB$ is the problem of distinguishing $\logic$-sentences $\phi$ that are true in $\rel A$ from those that are not true in $\rel B$. This problem makes sense only if every $\logic$-sentence that is true in $\rel A$ is also true in $\rel B$; we call such pairs $\LPMC$ templates.

\begin{definition}
A pair of similar structures $\AB$ is called an \emph{$\LPMC$ template} if $\rel A \vDash \phi$ implies $\rel B \vDash \phi$  for every $\logic$-sentence $\phi$ in the signature of $\rel A$ and $\rel B$.

Given an $\LPMC$ template $\AB$, the \emph{$\logic$-Promise Model Checking Problem over $\AB$}, denoted $\LPMC \AB$, is the following problem.  

Input: an $\logic$-sentence $\phi$ in the signature of $\rel A$ and $\rel B$;

Output: $\yes$ if $\rel A \vDash \phi$; $\no$ if $\rel B \not\vDash \phi$.
\end{definition}

The definition of a template guarantees that the sets of $\yes$-instances and $\no$-instances are disjoint. However, their union need not be the whole set of $\logic$-sentences; an algorithm for $\LPMC$ is only required to produce correct outputs for $\yes$-instances and $\no$-instances. Alternatively, we are \emph{promised} that the input sentence is a $\yes$-instance or a $\no$-instance. The complexity-theoretic notions (such as membership in $\NP$, $\NP$-completeness, reductions) can be  adjusted naturally for the promise setting. 
We write $\LPMC\CD \leq \LPMC \AB$ if the former problem can be reduced to the latter problem by a logarithmic space reduction, that is, a logarithmic space transformation that maps each $\yes$-instance $\phi$ of $\LPMC \CD$ to a $\yes$-instance $\psi$ of $\LPMC \AB$ (equivalently, $\rel C \vDash \phi$ must imply $\rel A \vDash \psi$) and $\no$-instances to $\no$-instances (equivalently, $\rel B \vDash \psi$ must imply $\rel D \vDash \phi$). 

An appropriate adjustment of definability for the promise setting is as follows. Note that we do not allow the negation in $\logic$, otherwise the concept would need to be defined differently because of the inclusions in the definition.

\begin{definition} \label{def:p-definability}
Assume $\neg \not\in \logic$ and 
let $\AB$ be a pair of similar structures. We say that a pair of relations $(S,T)$, where $S \subseteq A^n$ and $T \subseteq B^n$, is \emph{promise-$\logic$-definable} (or \emph{p-$\logic$-definable}) from $\AB$ if there exist relations $S'$ and $T'$ and an $\logic$-formula $\psi(v_1, \dots, v_n)$ such that $S \subseteq S'$, $T' \subseteq T$,
$\psi(v_1, \dots, v_n)$ defines $S'$ in $\rel A$, and
$\psi(v_1, \dots, v_n)$ defines $T'$ in $\rel B$.

We say that an $\LPMC$ template  $\CD$ is p-$\logic$-definable from $\AB$ (the signatures can differ) if $(Q^{\rel C},Q^{\rel D})$ is p-$\logic$-definable from $\AB$ for each relation symbol $Q$ in the signature of $\rel C$ and $\rel D$. 
\end{definition}

\begin{theorem} \label{thm:reduction}
Assume $\neg \not\in \logic$. If $\AB$ and $\CD$ are $\LPMC$ templates such that $\CD$ is p-$\logic$-definable from $\AB$, then $\LPMC \CD \leq \LPMC \AB$.
\end{theorem}

\begin{proof}
   The reduction is to replace each atomic $Q(\tuple{v})$ by the corresponding formula $\psi$ from Definition \ref{def:p-definability}. 
   For correctness of this reduction, observe that an $\logic$-sentence which is true in a structure $\rel E$ remains true when we add tuples to the relations of $\rel E$ (since $\logic$ does not contain $\neg$). 
\end{proof}

\subsection{Interesting fragments} \label{subsec:interesting-fragments}

We now explain why only four fragments of first-order logic need to be considered in order to fully understand the problems $\LPMC \AB$.
Observe first that if $\logic$ does not contain any connective ($\land,\lor$), or $\logic$ does not contain any quantifier ($\exists, \forall$), or $\logic \subseteq \{\exists, \lor\}$, then each $\LPMC$ is in $\Lspace$, the logarithmic space. (In some of these cases we do not even have any valid inputs in our definition of structures.) 

Secondly, notice that $(\logic \cup \{=\})$-$\PMC \AB$  is essentially the same as $\LPMC (\rel A',\rel B')$, where $\rel A'$ and $\rel B'$ are obtained from the original structures by adding a fresh binary symbol $Q$ to the signature and setting $Q^{\rel A'}$ to $=_A$ and $Q^{\rel B'}$ to $=_B$. The disequality is dealt with analogously, thus we can and shall restrict to fragments with $\logic \subseteq \efaoneg$.

Next, we deal with the negation. If $\neg$ is in $\logic$, and $\logic$ contains a quantifier and a connective, then it is enough to consider the case $\logic = \efaoneg$ since the remaining quantifier and connective can be expressed using negation. Moreover, the complements of relations can also be expressed, so we may assume that each template $\AB$ is \emph{closed under complementation}, meaning that for every symbol $R$ in the signature, we have a symbol $\compl{R}$ interpreted as $\compl{R}^{\rel A} = \compl{R^{\rel A}}$, $\compl{R}^{\rel B} = \compl{R^{\rel B}}$. But then $\neg$ is no longer necessary since we can propagate the negations inwards in an input sentence. We are down to $\logic \subseteq \efao$.

Finally, note that $\rel E \vDash \neg\phi$, where $\phi$ is an $\logic$-sentence, is equivalent to $\compl{\rel E} \vDash \phi'$ where $\phi'$ is an $\logic'$-sentence and $\logic'$ is obtained from $\logic$ by swapping $\forall \leftrightarrow \exists$ and $\lor \leftrightarrow \land$ ($\phi'$ can be, again, computed from $\neg \phi$ by inward propagation). It follows that $\phi \mapsto \phi'$ transforms every $\yes$-instance (resp., $\no$-instance) of $\LPMC \AB$ to a $\no$-instance (resp., $\yes$-instance) of $\logic'$-$\PMC (\compl{\rel B},\compl{\rel A})$, and a similar ``dual'' reduction works in the opposite direction. Therefore, the latter $\PMC$ has the ``dual'' complexity to the former $\PMC$, e.g., if the former is $\NP$-complete, then the latter is $\coNP$-complete; and if the former is $\PSPACE$-complete, then the latter is $\PSPACE$-complete as well. We will refer to this reasoning as the \emph{duality argument}.

Eliminating one of the logic fragments from each of the ``dual'' pairs, we are left with only four fragments:  $\logic = \ea$ (whose $\logic$-PMC is Promise CSP), $\logic = \efa$ (Promise Quantified CSP), $\logic = \eao$, and $\logic = \efao$. We investigate the last two separately in the next two sections.

\section{Existential positive fragment}

This section concerns the existential positive equality-free logic, that is, the $\logic$-logic with $\logic = \eao$. We fix this $\logic$ for the entire section.
\subsection{Characterization of templates and p-$\logic$-definability}

We start by characterizing $\LPMC$ templates. One direction of the characterization follows from the discussion below (\ref{eq:endo}), the other one from the following observation.

\begin{lemma} \label{lem:homs-formulas}
Let $f$ be a multi-homomorphism from $\rel A$ to $\rel B$, let $\phi(x_1, \dots, x_n)$ be a quantifier-free $\logic$-formula in the same signature, and let $\tuple{a} \in A^n$, $\tuple{b} \in B^n$. If $\rel A \vDash \phi(\tuple{a})$ and $\tuple{b} \in f(\tuple{a})$, then $\rel B \vDash \phi(\tuple{b})$. 
\end{lemma}

\begin{proof}
The claim holds for atomic formulas by definition of multi-homomorphisms. The proof is then finished by induction on the complexity of $\phi$; both $\lor$ and $\land$ are dealt with in a straightforward way.
\end{proof}

\begin{proposition} \label{prop:existential_templates}
A pair $\AB$ of similar structures is an $\LPMC$ template if and only if  there exists a homomorphism from $\rel A$ to $\rel B$. \end{proposition}

\begin{proof}
Suppose that there exists a homomorphism from $\rel A$ to $\rel B$ and $\rel A \vDash \phi$, where $\phi = \exists x_1 \exists x_2 \dots \exists x_n \phi'(x_1, \dots, x_n)$ is in prenex normal form. 
Then we have $\rel A \vDash \phi'(\tuple{a})$ for some $\tuple{a} \in A^n$, therefore $\rel B \vDash \phi'(f(\tuple{a}))$ by Lemma \ref{lem:homs-formulas}, and it follows that $\rel B \vDash \phi$.

For the forward implication, observe that the sentence obtained from the formula~(\ref{eq:endo}) by existentially quantifying all the variables is true in $\rel A$ (as there exists a homomorphism from $\rel A$ to $\rel A$ -- the identity), so it must be true in $\rel B$, giving us a homomorphism from $\rel A$ to $\rel B$. 
\end{proof} 

Note that this characterization would remain the same if we add $=$ to $\logic$ (and/or remove $\vee$).
For the following characterization of promise definability, the absence of the equality relation does make a difference, which is why we need to use multi-homomorphisms instead of homomorphisms.   

\begin{theorem} \label{thm:eao_definability}
  Let $\AB$ and $\CD$ be $\LPMC$ templates such that $A=C$ and $B=D$. Then $\CD$ is p-$\logic$-definable from $\AB$  if and only if $\MultHom \AB \subseteq \MultHom \CD$.  Moreover, in such a case, $\LPMC \CD \leq \LPMC \AB$.
\end{theorem}

\begin{proof}
It is enough to verify the equivalence, since then the second claim follows from Theorem~\ref{thm:reduction}. To prove the forward implication, assume that $(\rel C,\rel D)$ is p-$\logic$-definable from $\AB$, let $f \in \MultHom(\rel A,\rel B)$, and let $Q$ be a symbol in the signature of $\rel C$ and $\rel D$. To show that $f(\tuple{a}) \subseteq Q^{\rel D}$ for any $\tuple{a} \in Q^{\rel C}$ we apply Lemma~\ref{lem:homs-formulas} as follows. We have $\rel A \vDash \psi(\tuple{a})$, where $\psi(\tuple{x}) = \exists y_1 \exists y_2 \dots \exists y_m \psi'(\tuple{x},\tuple{y})$ is a formula from Definition~\ref{def:p-definability}, turned into prenex normal form. Then $\rel A \vDash \psi'(\tuple{a},\tuple{a}')$ for some $\tuple{a}' \in A^m$, thus $\rel B \vDash \psi'(\tuple{b}, \tuple{b'})$ for any $\tuple{b} \in f(\tuple{a})$ and $\tuple{b}' \in f(\tuple{a}')$ by Lemma~\ref{lem:homs-formulas}. Therefore, $\rel B \vDash \psi(\tuple{b})$ and, finally, $\tuple{b} \in Q^{\rel D}$, as required.

For the backward implication, assume that $\MultHom \AB \subseteq \MultHom \CD$, denote $\sigma$  the signature of $\rel A$ and $\rel B$, and consider an $n$-ary relational symbol $Q$ in the signature of $\rel C$ and $\rel D$. To prove the claim, we need to find a formula $\psi(x_1, \dots, x_n)$ that defines, in $\rel A$, a relation containing $Q^{\rel C}$ and, in $\rel B$, a relation contained in $Q^{\rel D}$. 

For simplicity, assume $A = [k]$ and consider the formula
\begin{equation} \label{eq:multi_homo}
  \phi(x_{1,1}, \dots, x_{1,n}, x_{2,1}, \dots, x_{2,n}, \dots, x_{k,n}) 
  :=     \bigwedge_{R \in \sigma} 
         \bigwedge_{\tuple{r} \in R^{\rel A}}  
         \bigwedge_{\tuple{j} \in [n]^{\ar(R)}}
         R(x_{r_1,j_1}, \dots, x_{r_{\ar(R)},j_{\ar(R)}})
\end{equation}
It follows immediately from definitions that, for any structure $\rel E$ in the signature $\sigma$, we have $\rel E \vDash \phi(e_{1,1}, \dots, e_{k,n})$ if and only if the mapping $i \mapsto \{e_{i,1}, \dots, e_{i,n}\}$, $1 \leq i \leq k$ is a multi-homomorphism from $\rel A$ to $\rel E$. Therefore, for any $\tuple{a} \in A^n$, the formula $\tau_{\tuple{a}}(x_1, \dots, x_n)$, obtained from $\phi$ by renaming $x_{a_i,i}$ to $x_i$ and existentially quantifying the remaining variables, defines in $\rel E$ the 
union of $f(\tuple{a})$ over $f \in \MultHom(\rel A,\rel E)$ of multiplicity at most $n$. This relation is clearly equal to the union of $f(\tuple{a})$ over all $f \in \MultHom(\rel A,\rel E)$. The sought after formula $\psi$ is then the disjunction of $\tau_{\tuple{a}}$ over all $\tuple{a} \in Q^{\rel C}$: it defines in $\rel A$ a relation containing $Q^{\rel C}$ (because of the identity ``multi''-homomorphism $\rel A \to \rel A$) and, in $\rel B$, a relation contained in $Q^{\rel D}$ (because every multi-homomorphism from $\rel A$ to $\rel B$ is a multi-homomorphism from $\rel C$ to $\rel D$, whence $f(\tuple{a}) \subseteq Q^{\rel D}$ for any $\tuple{a} \in Q^{\rel C}$ and any $f \in \MultHom \AB$).
\end{proof}

\subsection{Complexity classification}

Since $\LPMC \AB$ reduces to $\LMC (\rel A)$ (or $\LMC (\rel B)$) by the trivial reduction which does not change the input, and the latter problem is clearly in $\NP$, then the former problem is in $\NP$ as well. Theorem~\ref{thm:eao_classification} shows that $\LPMC \AB$ is NP-hard in all the ``nontrivial'' cases, as in the non-promise setting. However, our proof of hardness requires (in addition to Theorem~\ref{thm:eao_definability}) a much more involved hardness result than in the non-promise case: NP-hardness of $c$-coloring rainbow $k$-colorable $2k$-uniform hypergraphs from~\cite{Guruswami2018} (here $c,k \geq 2$).

To state the result in our formalism, we introduce the $n$-ary ``rainbow coloring'' and ``not-all-equal'' relations on a set $D$ as follows.
\begin{equation*}
    \rb_D^n = \{\tuple{d} \in D^n: \{d_1, d_2, \dots, d_n\} = D\}, \quad
    \nae_D^n = \{\tuple{d} \in D^n: \neg (d_1 = d_2 = \dots = d_n)\}
\end{equation*}
In the statement  of Theorem \ref{thm:hypergraph-hardness} and further, we use $(A; S_1, \dots, S_k)$ to denote a structure with universe $A$ and relations $S_1$, \dots, $S_k$.
\begin{theorem}[Corollary 1.2 in \cite{Guruswami2018}] \label{thm:hypergraph-hardness}
For any $A$ and $B$ of size at least 2, the problem $\ea$-$\PMC((A; \rb_A^{2|A|}),(B; \nae_B^{2|A|}))$ is NP-complete.
\end{theorem}

Given this hardness result, the complexity classification is a simple consequence of  Theorem~\ref{thm:eao_definability}.

\begin{theorem}[$\logic = \eao$] \label{thm:eao_classification}
Let $\AB$ be an $\LPMC$ template. If there is a constant homomorphism from $\rel A$ to $\rel B$, then $\LPMC \AB$ is in $\Lspace$ (in fact, decidable in constant time), otherwise $\LPMC \AB$ is $\NP$-complete.
\end{theorem}

\begin{proof}
If there exists a constant homomorphism $f: \rel A \to \rel B$, say with image $\{b\}$, then all the relations $R^{\rel B}$ in $\rel B$ contain the constant tuple $(b, b, \dots, b)$. It follows that every input sentence is satisfied in $\rel B$ by evaluating the existentially quantified variables to $b$; therefore, $\yes$ is always a correct output.

If there is no constant homomorphism $\rel A \to \rel B$, we observe that no multi-homomorphism from $\rel A$ to $\rel B$ contains a constant homomorphism (as the set of multi-homomorphisms of a $\PMC$ template is closed under containment).
It follows that the image of any ``rainbow'' tuple of $A$ under any multi-homomorphism from $\rel A$ to $\rel B$ does not contain any constant tuple, and so any multi-homomorphism from $\rel A$ to $\rel B$ is a multi-homomorphism 
from $(A; \rb_A^{2|A|})$ to $(B; \nae_B^{2|A|})$. The reduction from Theorem~\ref{thm:eao_definability} and the hardness from Theorem~\ref{thm:hypergraph-hardness} conclude the proof.
\end{proof}

\section{Positive fragment} \label{sec:positive}

We now turn our attention to the more complex case -- the positive equality-free logic, that is, the $\logic$-logic with $\logic = \efao$. We again fix this $\logic$ for the entire section.

\subsection{Witnesses for quantified formulas}

It will be convenient to work with $\logic$-formulas of the  special form
\begin{equation} \label{eq:prenex}
  \phi(x_1, \dots, x_n) = \forall y_1 \exists z_1 \forall y_2 \exists z_2 \ldots \forall y_m \exists z_m \ \phi'(\tuple{x},\tuple{y}, \tuple{z}),
\end{equation}
where $\phi'$ is quantifier-free. Note that each formula is equivalent to a formula in this form (by transforming to prenex normal form and adding dummy quantification as needed) and the conversion can be done in logarithmic space.

Observe that for a structure $\rel A$ and a tuple $\tuple{a} \in A^n$, we have $\rel A \vDash \phi(\tuple{a})$ if and only if 
there exist functions $\alpha_1: A \to A$, $\alpha_2: A^2 \to A$, \dots, $\alpha_m: A^m \to A$ which give us evaluations of the existentially quantified variables given the value of the previous universally quantified variables, i.e., these functions satisfy $\rel A \vDash \phi'(\tuple{a},\tuple{c},\alpha_1(c_1), \alpha_2(c_1, c_2), \dots, \alpha_m(c_1, \dots, c_m))$ for every $\tuple{c} \in A^m$. We call such functions \emph{witnesses} for $\rel A \vDash \phi(\tuple{a})$.

We state a simple consequence of this viewpoint, a version of Lemma~\ref{lem:homs-formulas}.

\begin{lemma} \label{lem:surj-homs-formulas}
Let $f$ be a surjective multi-homomorphism from $\rel A$ to $\rel B$, let $\phi(x_1, \dots, x_n)$ be an $\logic$-formula in the same signature as $\rel A$ and $\rel B$, and let $\tuple{a} \in A^n$, $\tuple{b} \in B^n$. If $\rel A \vDash \phi(\tuple{a})$ and $\tuple{b} \in f(\tuple{a})$, then $\rel B \vDash \phi(\tuple{b})$.

In particular, if there exists a surjective multi-homomorphism from $\rel A$ to $\rel B$, and $\phi$ is an $\logic$-sentence such that $\rel A \vDash \phi$, then $\rel B \vDash \phi$. 
\end{lemma}
\begin{proof}
The claim holds for quantifier-free $\logic$-formulas by Lemma~\ref{lem:homs-formulas}.

Next, we assume that $\phi$ is of the form (\ref{eq:prenex}) and select witnesses $\alpha_1$, \dots, $\alpha_m$ for $\rel A \vDash \phi(\tuple{a})$. Let $g: B \to A$ be any function such that $b \in f(g(b))$ for every $b \in B$, which exists as $f$ is surjective. We claim that any functions $\beta_1$, \dots, $\beta_m$ such that $\beta_i(b_1, \dots, b_i) \in f(\alpha_i(g(b_1), \dots, g(b_i)))$ for every $i \in [m]$, are witnesses for $\rel B \vDash \phi(\tuple{b})$. Indeed, for all $\tuple{d} \in B^m$, we have $\rel A \vDash \phi'(\tuple{a},g(\tuple{d}),\alpha_1(g(d_1)), \dots, \alpha_m(g(d_1), \dots, g(d_m)))$, and also $\tuple{b} \in f(\tuple{a})$, $\tuple{d} \in f(g(\tuple{d}))$, and $\beta_i(d_1, \dots, d_i) \in f(\alpha_i(g(d_1), \dots, g(d_i)))$ (by the assumption, choice of $g$, and choice of $\beta_i$, respectively); therefore,
$\rel B \vDash \phi'(\tuple{b},\tuple{d},\beta_1(d_1), \dots, \beta_m(d_1, \dots, d_m))$ by the first paragraph. 
\end{proof}

\subsection{Characterization of templates and p-$\logic$-definability}
Unlike in the existential case, both characterizations require  surjective and multi-valued functions. The core of these characterizations is an adjustment of (\ref{eq:multi_homo}) for surjective homomorphisms.

\begin{lemma} \label{lem:surj_hom}
  Let $\rel A$ be a structure with $A=[k]$ and $m,n$ be arbitrary positive integers. Then there exists a formula $\phi(x_{1,1}, \dots, x_{1,n}, x_{2,1}, \dots, \dots, x_{k,n})$ such that, for any structure $\rel E$ similar to $\rel A$ with $|E|\leq m$, we have $\rel E \vDash \phi(e_{1,1}, \dots, e_{k,n})$  if and only if the mapping $i \mapsto \{e_{i,1}, \dots, e_{i,n}\}$, $i \in [k]$ is contained in a surjective multi-homomorphism from $\rel A$ to $\rel E$. 
\end{lemma}

\begin{proof}
For every  function $h$ from $[m]$ to $[k]$  we take a formula $\phi_h(x_{1,1}, \dots, x_{k,n}, z_1, \dots, z_m)$ such that, for any structure $\rel E$ in the signature of $\rel A$, we have $\rel E \vDash \phi_h(e_{1,1}, \dots, e_{k,n}, e'_1, \dots, e'_m)$  if and only if the mapping $i \mapsto \{e_{i,1}, \dots, e_{i,n}\} \cup \bigcup_{h(l)=i}{e'_l}$, $1 \leq i \leq k$, is a multi-homomorphism from $\rel A$ to $\rel E$. Such a formula can be obtained by directly translating the definition of a multi-homomorphism into the language of logic, similarly to~(\ref{eq:multi_homo}).

We claim that the formula $\phi$ obtained by taking the disjunction of $\phi_{h}$ over all $h: [m] \to [k]$ and universally quantifying the variables $z_{1}$, \dots, $z_{m}$ satisfies the requirement of the lemma, provided $|E| \leq m$. Indeed, on the one hand, if $\rel E \vDash \phi(e_{1,1}, \dots, e_{k,n})$, then for every evaluation of the $z$ variables, some $\phi_{h}$ must be satisfied. We choose any evaluation that covers the whole set $E$ (which is possible since $|E| \leq m$) and the satisfied disjunct $\phi_{h}$ then gives us the required surjective multi-homomorphism from $\rel A$ to $\rel E$ (by the choice of $\phi_h$). 
On the other hand, if $i \mapsto \{e_{i,1}, \dots, e_{i,n}\}$ is contained in a surjective multi-homomorphism $f$, then for any evaluation $\eval_{E}(z_1)$, \dots, $\eval_{E}(z_m)$ of the universally quantified variables, a disjunct $\phi_h$ is satisfied whenever $\eval_{E}(z_l) \in f(h(l))$ for every $l \in [m]$. Such an $h$ exists since $f$ is surjective.
\end{proof}

\begin{proposition} \label{prop:positive_templates}
A pair $\AB$ of similar structures is an $\LPMC$ template if and only if  there exists a surjective multi-homomorphism from $\rel A$ to $\rel B$. \end{proposition}

\begin{proof}
For the forward implication, consider the sentence obtained by existentially quantifying all the variables in the formula $\phi$ provided by Lemma~\ref{lem:surj_hom} (with $m \geq |A|,|B|$). This sentence is true in $\rel A$ (as there exists a surjective multi-homomorphism from $\rel A$ to $\rel A$ -- the identity), so it must be true in $\rel B$, giving us a surjective multi-homomorphism from $\rel A$ to $\rel B$. 
The backward implication follows from Lemma~\ref{lem:surj-homs-formulas}.
\end{proof} 

An example which shows that one cannot replace in Proposition~\ref{prop:positive_templates} ``surjective multi-homomorphism'' by ``(multi-)homomorphism'' is the input formula $\varphi = \forall x \exists y R(x,y)$ (``there are no sinks'')
for a template where $\rel A$ is a digraph with no sinks and $\rel B$ is, say, $\rel A$ plus an isolated vertex. 

The following characterization of promise definability is also a straightforward consequence of Lemmata~\ref{lem:surj-homs-formulas} and \ref{lem:surj_hom}.

\begin{theorem} \label{thm:efao_definability}
  Let $\AB$ and $\CD$ be $\LPMC$ templates such that $A=C$ and $B=D$. Then $\CD$ is p-$\logic$-definable from $\AB$  if and only if $\SMultHom \AB \subseteq \SMultHom \CD$.  Moreover, in such a case, $\LPMC \CD \leq \LPMC \AB$.
\end{theorem}

\begin{proof}
The theorem is proved in the same way as Theorem~\ref{thm:eao_definability}; using Lemma~\ref{lem:surj-homs-formulas} instead of Lemma~\ref{lem:homs-formulas} for the forward implication, and the formula provided by Lemma~\ref{lem:surj_hom} instead of  (\ref{eq:multi_homo}) for the backward implication.
\end{proof}

\subsection{Membership} \label{subsec:positive-membership}

Clearly, every $\LMC$, as well as $\LPMC$, is in $\PSPACE$. We now give a generalization of the remaining membership results from~\cite{Madelaine2011tetrachotomy} using an appropriate generalization of ``A-shops'' and ``E-shops'' from that paper. We say that a surjective multi-homomorphism $f$ from $\rel A$ to $\rel B$ is an $\ashop$ if there exists
$a^* \in A$ such that $f(a^*) = B$. We also say that $\AB$ admits an $\ashop$ in such a case. 
We call $f$ an $\eshop$ if $f^{-1}(b^*) = A$ for some $b^* \in B$. Finally, we call $f$ an $\aeshop$ if it is simultaneously an $\ashop$ and an $\eshop$.

An additional simple reduction will be useful in the proof of the membership result (Theorem~\ref{thm:tractability}) and later as well. We say that an $\LPMC$ template $(\rel C,\rel D)$ is a \emph{relaxation} of an $\LPMC$ template $\AB$ if $(\rel C,\rel A)$ and $(\rel B,\rel D)$ are $\LPMC$ templates. Recall that, by Proposition~\ref{prop:positive_templates}, the property is equivalent to the existence of surjective multi-homomorphisms from $\rel C$ to $\rel A$ and from $\rel B$ to $\rel D$. 

\begin{proposition} \label{prop:relaxation}
  Let $\AB$ and $(\rel C,\rel D)$ be $\LPMC$ templates. If $(\rel C,\rel D)$ is a relaxation of $\AB$, then $\LPMC(\rel C,\rel D) \leq \LPMC \AB$.
\end{proposition}

\begin{proof} 
The trivial reduction, which does not change the input, works. Indeed, $\yes$-instances of $\LPMC \CD$ are $\yes$-instances of $\LPMC \AB$ since $(\rel C,\rel A)$ is an $\LPMC$ template, and $\no$-instances of $\LPMC \CD$ are $\no$-instances of $\LPMC \AB$ since $(\rel B,\rel D)$ is an $\LPMC$ template.
\end{proof}

\begin{theorem} \label{thm:tractability}
Let $\AB$ be an $\LPMC$ template. Then the following holds.
\begin{enumerate}
    \item If $\AB$ admits an $\ashop$, then $\LPMC \AB$ is in $\NP$.
    \item If $\AB$ admits an $\eshop$, then $\LPMC \AB$ is in $\coNP$.
    \item If $\AB$ admits an $\aeshop$, then $\LPMC \AB$ is in $\Lspace$.
\end{enumerate}
\end{theorem}

\begin{proof} 
For the first item, let $f$ be an $\ashop$ from $\rel A$ to $\rel B$ with $f(a^*) = B$, and consider an input $\phi$ in the special form (\ref{eq:prenex}), i.e., $\phi = \forall y_1 \exists z_1 \forall y_2 \exists z_2 \ldots \forall y_m \exists z_m \ \phi'(\tuple{y}, \tuple{z})$, where $\phi'$ is quantifier-free. We answer $\yes$ if there exists $\tuple{a} \in A^m$ such that $\rel A \vDash \phi'(a^*, a^*, \dots, a^*, \tuple{a})$; this can be clearly decided in $\NP$. It is clear that the answer is $\yes$ whenever $\phi$ is a $\yes$-instance of $\LPMC \AB$. On the other hand, if $\rel A \vDash \phi'(a^*, \dots, a^*, \tuple{a})$, then any functions $\beta_1: B \to B$, \dots, $\beta_m: B^m \to B$ such that $\beta_i(b_1, \dots, b_i) \in f(a_i)$ (for all $i \in [m]$ and $b_1, \dots, b_m \in B$) provide witnesses for $\rel B \vDash \phi$ by Lemma~\ref{lem:homs-formulas}. Therefore, if $\phi$ is a $\no$-instance, then the answer is $\no$, as needed. 

The second item follows by the duality argument.

In the case $\rel A = \rel B$, the third item can be proved in an analogous way (by eliminating both quantifiers instead of just one), see Corollary 9 in \cite{Madelaine2011tetrachotomy}. For the general case, we will construct $\rel C$ such that there is an $\aeshop$ from $\rel C$ to $\rel C$ and there are surjective multi-homomorphisms from $\rel A$ to $\rel C$ and from $\rel C$ to $\rel B$. Then $\AB$ will be a relaxation of $(\rel C,\rel C)$ by Proposition~\ref{prop:positive_templates}, and then membership of $\LPMC \AB$ in $\Lspace$ will follow from Proposition~\ref{prop:relaxation} and the mentioned Corollary 9 in \cite{Madelaine2011tetrachotomy}. Let $f$ be an $\aeshop$ from $\rel A$ to $\rel B$ with $f(a^*)=B$ and $f^{-1}(b^*)=A$, and define a surjective multi-valued function $f'$ from $A$ to $B$ by $f'(a^*) = B$ and $f'(a)=\{b^*\}$ if $a \neq a^*$. Note that $f'$ is contained in $f$, so $f'$ is a surjective multi-homomorphism from $\rel A$ to $\rel B$. We define $\rel C$ as the ``image'' of $\rel A$ under $f'$, that is, $C = B$ and $R^{\rel C} = \cup_{\tuple{a} \in R^{\rel A}} f'(\tuple{a})$ for each relation symbol $R$. Clearly, $f'$ is a surjective multi-homomorphism from $\rel A$ to $\rel C$ and the identity is a surjective homomorphism from $\rel C$ to $\rel B$. It remains to find an $\aeshop$ from $\rel C$ to $\rel C$. We claim that $g$ defined by $g(b^*) = \{b^*\}$ and $g(c)=C$ for $c \neq b^*$ is such an $\aeshop$. Indeed, if $\tuple{c} \in R^{\rel C}$, then $\tuple{c} \in f'(\tuple{a})$ for some $\tuple{a} \in R^{\rel A}$. By the definition of $f'$, we necessarily have $a_i=a^*$ whenever $c_i \neq b^*$; therefore, $f'(\tuple{a}) \supseteq g(\tuple{c})$. But $f'(\tuple{a}) \subseteq R^{\rel C}$ as $f' \in \SMultHom(\rel A,\rel C)$, and we are done.
\end{proof}

These membership results together with the (more involved) hardness results were sufficient for the tetrachotomy in~\cite{Madelaine2011tetrachotomy}. One problem with generalizing this tetrachotomy  is that, unlike in the non-promise setting, an $\LPMC$ template can admit an $\ashop$ and an $\eshop$, but no $\aeshop$. 
However, such a situation cannot happen for digraphs.

\begin{proposition} \label{prop:digraphs}
Let $\AB$ be an $\LPMC$ template such that $\rel A$ and $\rel B$ are digraphs. If $\AB$ admits an $\ashop$ and an $\eshop$, then it admits an $\aeshop$. 
\end{proposition}

\begin{proof}
See Appendix~\ref{append}.
\end{proof}

\subsection{Hardness}

As a consequence of Theorems~\ref{thm:hypergraph-hardness} and \ref{thm:efao_definability}, we obtain the following hardness result.

\begin{theorem} \label{thm:efao-hardness}
Let $\AB$ be an $\LPMC$ template.
\begin{enumerate}
    \item If there is no $\eshop$ from $\rel A$ to $\rel B$, then $\LPMC \AB$ is $\NP$-hard.
    \item If there is no $\ashop$ from $\rel A$ to $\rel B$, then $\LPMC \AB$ is $\coNP$-hard.
\end{enumerate}   
\end{theorem}

\begin{proof}
If there exists no $\eshop$ from $\rel A$ to $\rel B$, then $\SMultHom \AB$ is contained in $\SMultHom((A;\rb_A^{2|A|}),(B;\nae_B^{2|A|}))$.  Theorem~\ref{thm:hypergraph-hardness} and Theorem~\ref{thm:efao_definability} then imply the first item.
The second item follows by the duality argument.
\end{proof}

In the non-promise setting, the absence of $\ashop$s and $\eshop$s is sufficient for $\PSPACE$-hardness~\cite{Madelaine2011tetrachotomy,MadelaineTetrachotomyJournal}. This most involved part of the tetrachotomy result seems much more challenging in the promise setting and we do not have strong reasons to believe that templates without $\ashop$s and $\eshop$s are necessarily $\PSPACE$-hard. Nevertheless, we are able to prove some additional hardness results which will cover  all the extensions of $\logic$.

\begin{proposition} \label{prop:hardness_of_equality}
  $\LPMC((A; =_A),(B; =_B))$ is $\PSPACE$-hard for any $A$, $B$ such that $|A| \geq |B| \geq 2$. 
\end{proposition}

Note here that surjective multi-homomorphisms from $(A;=_A)$ to $(B;=_B)$ are exactly the surjective multi-valued functions from $A$ to $B$ of multiplicity one. In particular, if $|A|<|B|$, then $((A;=_A), (B;=_B))$ is not an $\LPMC$ template.
\begin{proof}
We start by noticing that the template $((A;=_A),([2]; =_{[2]}))$ is a relaxation of $(\rel A,\rel B) := ((A; =_A), (B;=_B))$. So by Proposition~\ref{prop:relaxation}, it is enough to prove the claim in the case $B = [2]$. For simplicity, we assume that $A = [k]$ ($k \geq 2$). We prove the $\PSPACE$-hardness by a reduction from $\LMC(\rel B)$, a $\PSPACE$-hard problem by, e.g.,~\cite{Martin2008modelchecking}. Consider an input $\phi$ to $\LMC(\rel B)$ in the special form~(\ref{eq:prenex}), i.e.,
$\phi = \forall y_1 \exists z_1 \forall y_2 \exists z_2 \ldots \forall y_m \exists z_m \ \phi'(\tuple{y}, \tuple{z})$,
where $\phi'$ is quantifier-free.
We need to find a log-space computable formula $\psi$ such that $\rel B \vDash \phi$ implies $\rel A \vDash \psi$ (so that $\yes$-instances of $\MC(\rel B)$ are transformed to $\yes$-instances of $\LPMC \AB$) and $\rel B \vDash \psi$ implies $\rel B \vDash \phi$ (so that $\no$-instances are transformed to $\no$-instances).

The rough idea to construct $\psi$ is to reinterpret the values in $A=[k]$ as values in $B = [2]$ via a mapping $A \to B$. We set
\begin{align}
\psi  &= \forall x_1 \forall x_2  \ \exists x_3 \exists x_4 \ldots \exists x_k \ \ (x_1=x_2) \ \vee \bigwedge_{f: A \to B} \rho_f, \quad  \mbox{ where } \label{eq:psi} \\
    \rho_f & = (\forall y_1' \exists z_1 \ldots \forall y'_m \exists z_m) \ \ (\exists y_1 \ldots \exists y_m) \ 
     \left(\bigwedge_{i=1}^{m} \sigma[f,y_i',y_i] \right)  \land \phi'(\tuple{y},\textbf{z}) \label{eq:rho}\\
    & \sigma[f,y_i',y_i] = \bigvee_{a \in A} \left((y_i' = x_a) \land (y_i = x_{f(a)})\right) \label{eq:sigma}
\end{align}

Observe first that $\psi$ can be constructed from $\phi$ in logarithmic space.

Next, we verify that $\rel B \vDash \psi$ implies $\rel B \vDash \phi$. So, we suppose $\rel B \vDash \psi$ and aim to find witnesses $\beta_1$, \dots, $\beta_m$ for $\rel B \vDash \phi$; to this end, let $\tuple{c}$ be some tuple in $B^m$ that corresponds to evaluations of universally quantified variables in $\phi$. We evaluate the variables $x_1$ and $x_2$ in $\psi$ as $\eval_{B}(x_1)=1$ and $\eval_{B}(x_2)=2$, and pick an evaluation $\eval_{B}(x_3), \dots, \eval_{B}(x_k)$ making $\psi$ true in $\rel B$. Set $f(a) = \eval_{B}(x_a)$, $a \in A$. The first disjunct of (\ref{eq:psi}) is not satisfied, so $\rho_f$ is satisfied with this choice of $\eval_{B}$. When it is the turn to evaluate $y_i'$, we set $\eval_{B}(y_i')=c_i$ and define $\beta_i(c_1, \dots, c_i)=\eval_{B}(z_i)$, where $\eval_{B}(z_i)$ is a satisfactory evaluation of $z_i$. Inspecting the definition (\ref{eq:sigma}), we see that $y_1$, \dots, $y_m$ are necessarily evaluated as $\eval_{B}(y_1)=c_1$, \dots, $\eval_{B}(y_m)=c_m$: indeed, if a disjunct $(y_i'=x_a) \land (y_i = x_{f(a)})$ is satisfied, then $c_i = \eval_{B}(y_i')=\eval_{B}(x_a)$ and $\eval_{B}(y_i) = \eval_{B}(x_{f(a)}) = \eval_{B}(x_{\eval_{B}(x_a)})=\eval_{B}(x_a)$; in particular, $\eval_{B}(y_i) = c_i$.
Therefore, the conjunct $\phi'(\tuple{y},\tuple{z})$ in (\ref{eq:rho}) ensures $\rel B \vDash \phi'(\tuple{c},\beta_1(c_1), \dots, \beta_m(c_1, \dots, c_m))$. As $\tuple{c}$ was chosen arbitrarily, we get that $\beta_1$, \dots, $\beta_m$ are witnesses for $\rel B \vDash \phi$, as required. 

We now suppose that $\beta_1$, \dots, $\beta_m$ are witnesses for $\rel B \vDash \phi$, and aim to show that $\rel A \vDash \psi$.
Because of the first disjunct of (\ref{eq:psi}), it is enough to  consider only evaluations of $x_1$ and $x_2$ with $\eval_{A}(x_1) \neq \eval_{A}(x_2)$. Since any bijection, regarded as a surjective multi-homomorphism from $\rel A$ to $\rel A$ of multiplicity one, preserves $\logic$-formulas (in the sense of Lemma~\ref{lem:surj-homs-formulas}), then we can as well assume that $\eval_{A}(x_1)=1$ and $\eval_{A}(x_2)=2$. We evaluate the remaining $x$ variables as $\eval_{A}(x_a)=a$, $a = 3,4, \dots, k$. We take a function  $f: A \to B$ and argue that $\rho_f$ is satisfied in $\rel A$. Given a selection of $\eval_{A}(y_i')$, we evaluate $z_i$ as $\eval_{A}(z_i) = \beta_i(f(\eval_{A}(y_1')), \dots, f(\eval_{A}(y_i')))$, and we define the evaluation of the remaining variables by $\eval_{A}(y_i)=f(\eval_{A}(y_i'))$. With these choices, each $\sigma[f,y_i',y_i]$ is satisfied because of the disjunct $a=\eval_{A}(y_i')$ in (\ref{eq:sigma}). The second conjunct in (\ref{eq:rho}), $\phi'(\tuple{y},\tuple{z})$, is also satisfied:
we know $\rel B \vDash \phi'(\tuple{c}, \beta_1(c_1), \dots, \beta_m(c_1, \dots, c_m))$ in particular for $c_1 = f(\eval_{A}(y_1'))$, \dots, $c_m = f(\eval_{A}(y_m'))$ and, with this $\tuple{c}$, it is apparent from the choice of evaluations that $\rel B \vDash \phi'(\tuple{c}, \beta_1(c_1), \dots, \beta_m(c_1, \dots, c_m))$ is equivalent to $\rel A \vDash \phi'(\eval_{A}(y_1), \dots, \eval_{A}(y_m), \eval_{A}(z_1), \dots, \eval_{A}(z_m))$. The proof of $\rel A \vDash \psi$ is concluded.
\end{proof}

It follows that $\efaoeq$-$\PMC$ over any template is $\PSPACE$-hard and so is, by the duality argument, $\efaoneq$-$\PMC$.
The next proposition implies $\PSPACE$-hardness for $\efaoneg$-$\PMC$.

\begin{proposition}
   Let $\AB$ be an $\LPMC$ template  which is closed under complementation. Then $\LPMC \AB$ is $\PSPACE$-hard.
\end{proposition}

\begin{proof} 
Suppose that $(\rel A,\rel B)$ is closed under complementation. We define an equivalence relation $\sim_A$ on $A$ by considering two elements equivalent if they play the same role in every relation of $\rel A$. Formally, $a \sim a'$ if for every symbol $R$ from the signature, every coordinate $i \in [\ar(R)]$, and every $\tuple{c}, \tuple{c}' \in A^{\ar(R)}$, if $c_i=a$, $c_i'=a'$, $c_j=c'_j$ for all $j \in [\ar(R)] \setminus \{i\}$, and $\tuple{c} \in R^{\rel A}$, then  $\tuple{c}' \in R^{\rel A}$. We define an equivalence relation $\sim_B$ on $B$ analogously. Notice that $\sim_A$ (resp., $\sim_B$) is indeed an equivalence relation; let $m$ and $n$ denote the number of equivalence classes of $\sim_A$ and $\sim_B$, respectively. 

Observe that $m,n \geq 2$. Indeed, otherwise any nonempty relation in the corresponding template contains all the tuples, and we do not allow such structures in this paper.

Let $\rel C = (A; \sim_A)$ and $\rel D = (B; \sim_B)$. We claim that every surjective multi-homomorphism $f$ from $\rel A$ to $\rel B$ preserves $\sim$, i.e., is a surjective multi-homomorphism from $\rel C$ to $\rel D$. Consider $a, a' \in A$, and $b,b' \in B$ such that $a \sim_A a'$, $b \in f(a)$, and $b' \in f(a')$. In order to prove $b \sim_B b'$, take arbitrary $R$, $i$, $\tuple{d}$, $\tuple{d}'$ such that $d_i=b$, $d_i'=b'$, $d_j=d'_j$ for all $j \neq i$, and $\tuple{d} \in R^{\rel B}$.
Let $\tuple{c}, \tuple{c}' \in A^{\ar(R)}$ be tuples such that $c_i=a$, $c'_i=a'$, and $c_j=c'_j \in f^{-1}(d_j)$ for all $j \neq i$ (which exist as $f$ is surjective). 
If $\tuple{c} \not\in R^{\rel A}$, then
$\tuple{c} \in \compl{R}^{\rel A}$ and, consequently, $\tuple{d} \in f(\tuple{c}) \subseteq \compl{R}^{\rel B}$ (as $f$ is a surjective multi-homomorphism from $\rel A$ to $\rel B$), a contradiction with $\tuple{d} \in R^{\rel B}$. Therefore, $\tuple{c} \in R^{\rel A}$ and also $\tuple{c}' \in R^{\rel A}$ as $a \sim_A a'$. Now $\tuple{d}' \in f(\tuple{c}') \subseteq R^{\rel B}$, and $b \sim_B b'$ follows. 

By Theorem~\ref{thm:efao_definability}, $\LPMC \CD \leq \LPMC \AB$. 
Since there exists a surjective multi-valued function from $A$ to $B$ that preserves $\sim$ (namely, any $f \in \SMultHom \AB$), we also know that $m \geq n$. The template $(\rel E,\rel F) := (([m]; =_{[m]}),([n]; =_{[n]}))$ is a relaxation of $\CD$, because there exists a surjective multi-homomorphism from $\rel E$
to $\rel C$ (a multi-valued function that maps $i$ to the $i$-th equivalence class of $\sim_A$ under an arbitrary linear ordering of classes)  and a surjective multi-homomorphism from $\rel D$ to $\rel F$
(a ``multi''-valued function that maps every element in the $i$-th equivalence class of $\sim_B$ to $\{i\}$).
By Proposition~\ref{prop:relaxation}, $\LPMC (\rel E,\rel F) \leq \LPMC \CD$; therefore, $\LPMC(\rel E,\rel F) \leq \LPMC \AB$. The former $\LPMC$ is $\PSPACE$-hard by Proposition~\ref{prop:hardness_of_equality}, so $\LPMC \AB$ is $\PSPACE$-hard, too.
\end{proof}

\subsection{Summary and examples} \label{subsec:examples}

The claims stated in Figure~\ref{fig:result} are now immediate consequences of the obtained results. Note that the claims remain  true without the imposed restrictions on structures (i.e., we can allow singleton universes, nullary relations, etc.); the only nontrivial ingredient is the $\Lspace$-membership of the Boolean Sentence Value Problem~\cite{Lynch77}.

We observe that the results imply  a complete complexity classification in the case that one of the two template structures is \emph{Boolean}, i.e., has a two-element universe.

\begin{corollary}[$\logic=\efao$]
Let $\AB$ be an $\LPMC$ template.
\begin{enumerate}
    \item If $\rel B$ is Boolean, then $\LPMC \AB$ is in $\Lspace$, or is $\NP$-complete, or $\PSPACE$-complete.
    \item If $\rel A$ is Boolean, then $\LPMC \AB$ is in $\Lspace$, or is $\coNP$-complete, or $\PSPACE$-complete.
   \item If $\rel A$ and $\rel B$ are Boolean, then $\LPMC \AB$ is  in $\Lspace$, or is $\PSPACE$-complete.
\end{enumerate}
\end{corollary}

\begin{proof}
If $\rel B$ is Boolean, then every $\eshop$ (from $\rel A$ to $\rel B$) is an $\aeshop$. Moreover, if there is no $\ashop$, then every surjective multi-homomorphism is of multiplicity one, so it is also a multi-homomorphism from $(A; =_A)$ to $(B; =_B)$. The first item now follows from Proposition~\ref{prop:hardness_of_equality} and Theorem~\ref{thm:efao_definability}. The other items are easy as well.
\end{proof}

There are two wide gaps left for further investigation. First, it is unclear what the complexity is for the $\LPMC$ over templates that admit both an $\ashop$ and an $\eshop$, but no $\aeshop$. While there is no such a digraph template, there are examples with one ternary or two binary relations, e.g., the following. We use $ij$ as a shortcut for the pair $(i,j)$.
\begin{align*}
\rel A &= ([3]; \ \{(1,2,3)\}), \quad
\rel B = ([3]; \ \{1,2,3\} \times \{2\} \times \{3\} \ \cup \ \{1,2\} \times \{2\} \times \{2,3\}) \\
\rel A &= ([3]; \ \{12\}, \ \{13\}), \quad
\rel B = ([3]; \ \{12,22,32\}, \ \{12,13,22,23,33\})
\end{align*}

The second gap is between simultaneous $\NP$- and $\coNP$-hardness, and $\PSPACE$-hardness, when the template admits neither an $\ashop$ nor an $\eshop$. Examples with unknown complexity include the following.
\begin{align*}
    \rel A &= ([3]; \ \{(1,2,3)\}), \quad
    \rel B = ([3]; \  \{2,3\} \times \{1,3\} \times \{1,2\}) \\
    \rel A &= ([3]; \ \{(1,2,3)\}), \quad \rel B = ([3]; \{1,2\} \times \{1,2\} \times \{3\} \ \cup \ \{1,3\} \times \{2\} \times \{2\}) \\
    \rel A &= ([4]; \ \{12,34\}), \quad
    \rel B = ([4]; \  \{12,13,14,23,24,34,32\})
\end{align*}
In an ongoing work, we have developed some more general $\PSPACE$-hardness criteria, but the examples above remain elusive. The following equivalent \emph{unary} version of the first example is an especially interesting template, whose $\LPMC$ is the problem described in the introduction.
\begin{align*}
    \rel A &= ([3]; \ \{1\}, \{2\}, \{3\}), \quad
    \rel B = ([3]; \ \{2,3\}, \{1,3\}, \{1,2\})
\end{align*}

\section{Conclusion}

We gave a full complexity classification of $\eao$-$\PMC$, initiated an algebraic approach to $\efao$-$\PMC$, and applied it to provide several complexity results about this class of problems. 

An interesting concrete problem, whose complexity is currently open, is the $\efao$-$\PMC$ over the unary template above. As for the theory-building, the next natural step is to capture more complex reductions by means of surjective multi-homomorphisms; namely, the analogue of pp-constructions, which proved to be so useful in the theory of (Promise) CSPs~\cite{Barto2017polymorphisms,Barto2021algebraicPCSP}. It may be also helpful to characterize and study the sets of surjective multi-homomorphisms in the spirit of \cite{MartinLattice,Carvalho2021}.

\bibliography{pmc}

\appendix 

\section{Proof of Proposition~\ref{prop:digraphs}}
\label{append}

Denote by $R$ the unique binary symbol in the signature. Let $f$ be an $\ashop$ from $\rel A$ to $\rel B$ with $f(a^*)=B$ and let $g$ be an $\eshop$ from $\rel A$ to $\rel B$ with $g^{-1}(b^*)=A$.

If $a^*$ is isolated in $\rel A$ (i.e., $(a,a^*), (a^*,a) \notin R^{\rel A}$ for every $a \in A$), then we define a surjective multi-valued function $h$ by $h(a^*)=B$ and $h(a)=\{b^*\}$ for every $a \neq a^*$. It is a multi-homomorphism from $\rel A$ to $\rel B$ since for any $(a,a') \in R^{\rel A}$, we have $h(a,a') = \{(b^*,b^*)\}$, which is contained in $R^{\rel B}$ because $R^{\rel A}$ is nonempty, so $g(R^{\rel A}) \ni (b^*,b^*)$.

Suppose next that there is an edge $(a_1,a^*) \in R^{\rel A}$ but $a^*$ has no outgoing edges in $\rel A$. Let $b_1$ be an arbitrary element from $f(a_1)$ and define $h$ by $h(a^*)=B$ and $h(a)=\{b_1\}$ for every $a \neq a^*$. To verify that $h \in \SMultHom \AB$, consider an edge $(a,a') \in R^{\rel A}$. As $a^*$ has no outgoing edges in $\rel A$, we get $a \neq a^*$, so $h(a) = \{b_1\}$. Now $h(a,a') \subseteq \{b_1\} \times B$, which is contained in $R^{\rel B}$ because $R^{\rel B} \supseteq f(a_1,a^*) \supseteq \{b_1\} \times B$.

If $a^*$ has an outgoing edge $(a^*,a_1) \in R^{\rel A}$ but no incoming edges, we proceed similarly, defining $h(a^*) = B$ and $h(a)=\{b_1\}$ for all $a \neq a^*$, where $b_1$ is an arbitrary element from $f(a_1)$.

Finally, suppose that $(a_1,a^*) \in R^{\rel A}$ and $(a^*,a_2) \in R^{\rel A}$ for some $a_1, a_2 \in A$. If there is an element $a_3 \in A$ with no outgoing (resp., incoming) edges, define $h$ by $h(a_3)=B$ and $h(a)=\{b'\}$ for all $a \neq a_3$, where $b'$ is an arbitrary element from $f(a_1)$ (resp., $f(a_2)$). If there is no such element $a_3$, then we define $h(a^*) = B$ and $h(a) = \{b^*\}$ for all $a \neq a^*$. Since $g$ is surjective, and every $a \in A$ has both an incoming and an outgoing edge, then $(b,b^*) \in R^{\rel B}$ and $(b^*,b) \in R^{\rel B}$ for all $b \in B$, therefore, $h\in \SMultHom \AB$.

The proof of Proposition~\ref{prop:digraphs} is concluded.

\end{document}